\documentclass[12pt,letterpaper]{article}
\usepackage{amsmath,amsfonts,amsthm,amssymb}

\theoremstyle{plain}

\numberwithin{equation}{section}
\newtheorem{thm}{Theorem}[section]
\newtheorem{lem}[thm]{Lemma}
\newtheorem{cor}[thm]{Corollary}

\newenvironment{exam}[1]{\medskip          
  \setlength{\rightmargin}{\leftmargin}
  {\noindent\textbf{Example #1}.\enspace}%
  }%

\newcommand{\complex}{{\mathbb C}}
\newcommand{\positive}{{\mathbb N}}
\newcommand{\real}{{\mathbb R}}
\newcommand{\ascript}{{\mathcal A}}
\newcommand{\bscript}{{\mathcal B}}
\newcommand{\cscript}{{\mathcal C}}

\newcommand{\escript}{{\mathcal E}}
\newcommand{\lscript}{{\mathcal L}}
\newcommand{\uscript}{{\mathcal U}}
\newcommand{\rmre}{\mathrm{Re\,}}
\newcommand{\rmtr}{\mathrm{tr}}
\newcommand{\rmspan}{\mathrm{span}}
\newcommand{\rmnull}{\mathrm{Null}}
\newcommand{\rmrank}{\mathrm{rank}}
\newcommand{\rmrange}{\mathrm{range}}
\newcommand{\rmcyl}{\mathrm{cyl}}
\newcommand{\fhat}{\widehat{f}}
\newcommand{\ghat}{\widehat{g}}
\newcommand{\hhat}{\widehat{h}}
\newcommand{\muhat}{\widehat{\mu}}
\newcommand{\chihat}{\widehat{\chi}}

\newcommand{\ab}[1]{\left|#1\right|}

\newcommand{\brac}[1]{\left\{#1\right\}}
\newcommand{\paren}[1]{\left(#1\right)}
\newcommand{\sqbrac}[1]{\left[#1\right]}
\newcommand{\elbows}[1]{{\left\langle#1\right\rangle}}
\newcommand{\ket}[1]{{\left|#1\right>}}
\newcommand{\bra}[1]{{\left<#1\right|}}

\def\REMARK{\noindent {\bf Remark \ }}

\def\journaldata#1#2#3#4{{\it #1\/}\phantom{--}{\bf #2$\,$:} $\!$#3 (#4)}

\begin{document}

\title{TWO-SITE\\QUANTUM RANDOM WALK
}
\author{S. Gudder\\ Department of Mathematics\\
University of Denver\\ Denver, Colorado 80208, U.S.A.\\
sgudder@math.du.edu\\
and\\
Rafael D. Sorkin\\
Perimeter Institute for Theoretical Physics\\
Waterloo, Ontario N2L 2Y5 Canada\\
rsorkin@perimeterinstitute.ca
}
\date{}
\maketitle

\begin{abstract}
We study the measure theory of a two-site quantum random walk.  The
truncated decoherence functional defines a quantum measure $\mu_n$ on
the space of $n$-paths, and the $\mu_n$ in turn induce a quantum measure
$\mu$ on the cylinder sets within the space $\Omega$ of untruncated
paths.  Although $\mu$ cannot be extended to a continuous quantum
measure on the full $\sigma$-algebra generated by the cylinder sets, an
important question is whether it can be extended to sufficiently many
physically relevant subsets of $\Omega$ in a systematic way. We begin an
investigation of this problem by showing that $\mu$ can be extended to a
quantum measure on a ``quadratic algebra'' of subsets of $\Omega$ that
properly contains the cylinder sets.  We also present a new
characterization of the quantum integral on the $n$-path space.
\end{abstract}

\section{Introduction}  
A two-site quantum random walk is a process that describes the motion of
a quantum particle that occupies one of two sites $0$ and $1$. We assume
that the particle begins at site $0$ at time $t=0$ and either remains at
its present site or moves to the other site at discrete time steps
$t=0,1,2,\ldots\,$. The transition amplitude is given by the unitary
matrix
\begin{equation*}
U=\frac{1}{\sqrt{2\,}}\,
\left[\begin{matrix}\noalign{\smallskip}1&i\\\noalign{\smallskip}i&1\\\noalign{\smallskip}\end{matrix}\right]
\end{equation*}
Thus, the amplitude that the particle remains at its present position at
one time step is $1/\sqrt{2\,}$ and the amplitude that it changes
positions at one time step is $i/\sqrt{2\,}$. We can also interpret this
process as a quantum coin for which 0 and 1 are replaced by T (tails)
and H (heads), respectively.

This two-site process is a special case of a finite unitary system
\cite{djsu10} in which more than two sites are considered. Although we
present a special case we study the process in much greater detail, 
which
we believe gives more insight into the situation. We expect that some
of the work presented here will generalize to finite unitary systems.
Moreover the methods employed will be general enough to cover
non-unitary processes, which are ubiquitous for open systems and which
plausibly include the case of quantum gravity (cf. \cite{mocs05}).
We believe this greater generality is instructive, although it does
lengthen the derivations in some instances. 

Unlike previous studies of quantum random walks, the present work
emphasizes aspects of quantum measure theory \cite {gt0906, gud09,
  gud101, mocs05, sor94, sor071}. We begin by introducing the
$n$-truncated decoherence functional $D_n$ on the $n$-path space $\Omega
_n$ corresponding to $U$. The functional $D_n$ is then employed to
define a quantum measure $\mu _n$ on the events in $\Omega _n$. We then
use $\mu _n$ to define a quantum measure $\mu$ on the algebra of
cylinder sets $\cscript (\Omega )$ for the path space $\Omega$. Although
$\mu$ cannot be extended to a continuous quantum measure on the
$\sigma$-algebra generated by $\cscript (\Omega )$ \cite{djsu10}, an
important problem is whether $\mu$ can be extended to other physically
relevant sets in a systematic way. We begin an investigation of this
problem by introducing the concept of a quadratic algebra of sets. It is
shown that $\mu$ extends to a quantum measure on a quadratic algebra
that properly contains $\cscript (\Omega )$. 

We also consider a quantum integral with respect to $\mu_n$ of random
variables 
(real-valued functions)
on $\Omega _n$ \cite{gud09}. A new characterization of the
quantum integral $\int fd\mu _n$ is presented. It is shown that a random
variable 
$f\colon\Omega _n\to\real$ corresponds to a self-adjoint operator
$\fhat$ on a $2^n$-dimensional Hilbert space 
$H_n$ such that
\begin{equation*}
\int fd\mu _n=\rmtr (\fhat D^n)
\end{equation*}
where $D^n$ is a density operator on $H_n$ corresponding to $D_n$.

\section{Truncated Decoherence Functional} 
The sample space (or ``history-space'')
$\Omega$ consists of all sequences of zeros and ones beginning with zero. For example
$\omega\in\Omega$ with
\begin{equation*}
\omega = 0110110\cdots
\end{equation*}
We call the elements of $\Omega$ \textit{paths}. A finite string
\begin{equation*}
  \omega =\alpha _0\alpha _1\cdots\alpha _n,\ \alpha _i\in\brac{0,1},\ \alpha _0=0
\end{equation*}
is called an $n$-\textit{path}. 
If
\begin{equation*}
  \omega '=\alpha '_0\alpha '_1\cdots\alpha '_n
\end{equation*}
is another $n$-path, the 
\textit{joint amplitude} or ``Schwinger amplitude'' 
$D^n(\omega ,\omega ')$ between $\omega$ and $\omega '$ is
\begin{equation}         
  \label{eq21}
  D^n(\omega ,\omega ')=\tfrac{1}{2^n}\,i^{\ab{\alpha _1-\alpha _0}}\cdots i^{\ab{\alpha _n-\alpha _{n-1}}}
  i^{-\ab{\alpha '_1-\alpha '_0}}\cdots i^{-\ab{\alpha '_n-\alpha '_{n-1}}}\delta _{\alpha _n\alpha '_n}
\end{equation}
We call the set $\Omega_n$ of $n$-paths the $n$-\textit{path space} and write
\begin{equation*}
  \Omega _n=\brac{\omega _0,\omega _1,\ldots\,,\omega _{2^n-1}}
\end{equation*}
where 
$\omega_0=0\cdots 0$, 
$\omega_1=0\cdots 01$, 
$\omega_2=0\cdots 010$, 
$\ldots$,
$\omega_{2^n-1}=011\cdots 1$. 
Thus, $\omega_i=i$ in binary notation, $i=0,1,\ldots ,2^n-1$ and we can write
$\Omega _n=\brac{0,1,\ldots ,2^n-1}$. The $n$-\textit{truncated decoherence matrix}
(or $n$-\textit{decoherence matrix}, for short) is the $2^n\times 2^n$ matrix $D^n$ given by
\begin{equation*}
D_{ij}^n=D^n(\omega _i,\omega _j)=D^n(i,j)
\end{equation*}
The algebra of subsets of $\Omega _n$ is denoted by $\ascript _n$ or $2^{\Omega _n}$. The
$n$-\textit{decoherence functional} $D_n\colon\ascript _n\times\ascript _n\to\complex$ is defined by
\begin{equation*}
D_n(A,B)=\sum\brac{D_{ij}^n\colon\omega _i\in A,\omega _j\in B}
\end{equation*}
The $n$-\textit{truncated} $q$-\textit{measure} $\mu _n\colon\ascript _n\to\real ^+$ is defined by
$\mu _n(A)=D_n(A,A)$ (we shall shortly show that $\mu _n(A)\ge 0$ for all $A\in\ascript _n$). We then have
\begin{equation*}
\mu _n(A)=\sum\brac{D_{ij}^n\colon\omega _i,\omega _j\in A}
\end{equation*}
For $i,j=0,1,\ldots ,2^n-1$ $i\ne j$ we define the \textit{interference term}
\begin{equation*}
  I_{ij}^n =  \mu_n\paren{\brac{\omega_i,\omega_j}} - \mu_n(\omega_i) - \mu _n(\omega_j)
\end{equation*}
Since
\begin{align*}
\mu _n\paren{\brac{\omega _i,\omega _j}}&=D_n\paren{\brac{\omega _i,\omega _j},\brac{\omega _i,\omega _j}}
  =D_{ii}^n+D_{jj}^n+2\rmre D_{ij}^n\\
  &=\mu _n(\omega _i)+\mu _n(\omega _j)+2\rmre D_{ij}^n
\end{align*}
we have that
\begin{equation*}
   I_{ij}^n = 2 \rmre D_{i,j}^n = D_n(\omega_i,\omega_j)+D_n(\omega_j,\omega_i)
\end{equation*}
If $I_{ij}^n=0$ we say that $i$ and $j$ \textit{do not interfere}\footnote%
{In some contexts, the stronger condition $D_n(\omega_i,\omega_j)=0$
  would be more appropriate.}
and we write $i \, n \, j$; 
if $I_{ij}^n>0$, then $i$ and $j$ \textit{interfere constructively} 
and we write $i\,c\,j$; 
if $I_{ij}^n<0$, then $i$ and $j$ \textit{interfere destructively} 
and we write $i\,d\,j$.

\begin{exam}{1}                 
For $n=1$, $\Omega _1=\brac{00,01}$ and
\begin{equation*}
D^1=\frac{1}{2}\,
\left[\begin{matrix}\noalign{\smallskip}1&0\\\noalign{\smallskip}0&1\\\noalign{\smallskip}\end{matrix}\right]
\end{equation*}
We have $\mu _1(\emptyset )=0$, $\mu _1(\omega _0)=\mu _1(\omega _1)=1/2$, $\mu _1(\Omega _1)=1$. There is no interference and $\mu _1$ is a measure
\end{exam}

\begin{exam}{2}                    
For $n=1$, $\Omega _2=\brac{000,001,010,011}=\brac{0,1,2,3}$ and
\begin{equation*}
D^2=\frac{1}{4}\,
\left[\begin{matrix}\noalign{\smallskip}1&0&-1&0\\\noalign{\smallskip}
  0&1&0&1\\\noalign{\smallskip}-1&0&1&0\\\noalign{\smallskip}0&1&0&1\\\noalign{\smallskip}
  \end{matrix}\right]
\end{equation*}
We have $\mu _1(\emptyset )=0$, $\mu _2(i)=1/4$, $i=0,1,2,3$, $\mu\paren{\brac{0,2}}=0$
\begin{align*}
\mu _2\paren{\brac{0,1}}&=\mu _2\paren{\brac{0,3}}=\mu _2\paren{\brac{1,2}}=\mu _2\paren{\brac{2,3}}=1/2\\
\mu _2\paren{\brac{1,3}}&=1,\ \mu _2\paren{\brac{0,1,2}}=1/4\\
\mu _2\paren{\brac{0,1,3}}&=\mu _2\paren{\brac{1,2,3}}=5/4,\ \mu _2(\Omega _2)=1
\end{align*}
In this case there is interference and $\mu _2$ is not a measure. 
The interference terms are 
$I_{02}^2=-1/2$,
$I_{13}^2= 1/2$ and 
$I_{ij}^2=0$ for $i<j$, $(i,j)\ne (0,2),(1,3)$.

Let $c_n(\omega )$ be the number of position changes for an $n$-path $\omega$. For example, $c_4(01011)=3$
and $c_5(011010)=4$. If $\omega ,\omega '$ are $n$-paths it follows from \eqref{eq21} that
\begin{equation}         
\label{eq22}
D^n(\omega ,\omega ')=\tfrac{1}{2^n}\,i^{\sqbrac{c_n(\omega )-c_n(\omega ')}}\delta _{\alpha _n\alpha '_n}
\end{equation}
If two integers are both even or both odd, they have the \textit{same parity} and otherwise they have
\textit{different parity}. If $A=\sqbrac{a_{ij}}$ and $B=\sqbrac{b_{ij}}$ are $n\times x$ matrices, their
\textit{Hadamard product} $A\circ B$ is $\sqbrac{a_{ij}b_{ij}}$; that is, the $ij$-entry of $A\circ B$ is $a_{ij}b_{ij}$.
\end{exam}

\begin{thm}       
\label{thm21}
If $D^n$ is the $n$-truncated decoherence matrix, then $D^n$ is positive
semi-definite, $D_{jk}^n=0$ if $j,k$ have different parity and if $j,k$
have the same parity then $D_{jk}^n=1/2^n$ when $c_n(j)=c_n(k)\pmod{4}$
and $D_{jk}^n=-1/2^n$ when $c_n(j)\ne c_n(k)\pmod{4}$. Moreover, $\sum
_{j,k}D_{jk}^n=1$. 
\end{thm}
\begin{proof}
It is well-known that the Hadamard product of two positive semi-definite
square matrices of the same size is again positive semi-definite. It
follows from \eqref{eq22} that 
\begin{equation}         
\label{eq23}
D_{jk}^n=\tfrac{1}{2^n}\,i^{\sqbrac{c_n(j)-c_n(k)}}p_{jk}
\end{equation}
where $p_{jk}=1$ if $j,k$ have the same parity and $p_{jk}=0$, otherwise. Defining the matrices $P=\sqbrac{p_{jk}}$ and
\begin{equation*}
C=\sqbrac{i^{\sqbrac{c_n(j)-c_n(k)}}}
\end{equation*}
we have that $D^n=\tfrac{1}{2^n}\,C\circ P$. Now $C$ is clearly positive semi-definite. To show that $P$ is positive
semi-definite, notice that
\begin{equation*}
P=
\left[\begin{matrix}\noalign{\smallskip}1&0&1&0&\cdots&1&0\\\noalign{\smallskip}
  0&1&0&1&\cdots&0&1\\\noalign{\smallskip}1&0&1&0&\cdots&1&0\\\noalign{\smallskip}
  &\vdots&&&&&\\\noalign{\smallskip}0&1&0&1&\cdots&0&1\\\noalign{\smallskip}
  \end{matrix}\right]
\end{equation*}
We see that $P$ is self-adjoint, $\rmrank (P)=2$ and $\rmrange (P)$ is generated by the vectors
$v_1=(1,0,1,0,\ldots ,1,0)$ and $v_2=(0,1,0,1,\ldots ,0,1)$. Now $Pv_1=2^{n-1}v_1$ and $Pv_2=2^{n-1}v_2$. For any $2^n$-dimensional vector $v$ with $v\perp v_1$ and $v\perp v_2$, we have that $Pv=0$. Hence, the eigenvalues of $P$ are $0$ and $2^{n-1}$. It follows that $P$ is positive semi-definite and hence, $D^n$ is positive semi-definite. The values of $D_{jk}^n$ given in the statement of the theorem follow from \eqref{eq23}. By symmetry, there are as many $1$s as $-1$s among the off-diagonal entries of $D^n$. Hence, 
\begin{equation*}
\sum _{j,k=0}^{2^n-1}D_{jk}^n=\sum _{j=0}^{2^n-1}D_{jj}^n=\sum _{j=0}^{2^n-1}\frac{1}{2^n}=1
\qedhere
\end{equation*}
\end{proof}

It follows from Theorem~\ref{thm21} that $D^n$ is a density matrix.

\begin{exam}{3}                    
For $n=3$ we have $\Omega _3=\brac{0,1,\ldots ,7}$ and using vector notation $c_3=\paren{c_3(0),\ldots ,c_3(7)}$ we have that
\begin{equation*}
c_3=(0,1,2,1,2,3,2,1)
\end{equation*}
Applying Theorem~\ref{thm21} we can read off the entries of $D^3$ to obtain
\begin{equation*}
D^3=\frac{1}{8}
\left[\begin{matrix}\noalign{\smallskip}\ 1&\ 0&-1&\ 0&-1&\ 0&-1&\ 0\\\noalign{\smallskip}
  \ 0&\ 1&\ 0&\ 1&\ 0&-1&\ 0&\ 1\\\noalign{\smallskip}-1&\ 0&\ 1&\ 0&\ 1&\ 0&\ 1&\ 0\\\noalign{\smallskip}
  \ 0&\ 1&\ 0&\ 1&\ 0&-1&\ 0&\ 1\\\noalign{\smallskip}-1&\ 0&\ 1&\ 0&\ 1&\ 0&\ 1&\ 0\\\noalign{\smallskip}
  \ 0&-1&\ 0&-1&\ 0&\ 1&\ 0&-1\\\noalign{\smallskip}-1&\ 0&\ 1&\ 0&\ 1&\ 0&\ 1&\ 0\\\noalign{\smallskip}
  \ 0&\ 1&\ 0&\ 1&\ 0&-1&\ 0&\ 1\\\noalign{\smallskip}
  \end{matrix}\right]
\end{equation*}
\end{exam}

\begin{cor}       
\label{cor22}
{\rm (a)}\enspace $inj,jnk\Rightarrow i\!\!\not{\!n}k$; $inj,jdk$ or $jck\Rightarrow ink$.\newline
{\rm (b)}\enspace $icj,jck\Rightarrow ick$.
{\rm (c)}\enspace $idj,jdk\Rightarrow ick$.
{\rm (d)}\enspace $icj,jdk\Rightarrow idk$.
\end{cor}

\begin{exam}{4}                    
Referring to $D^3$ in Example~3 we see that $0d2$ $2c4$ and $0d4$. Also, $2d0$, $0d4$ and $2c4$. Finally, $1c3$, $3c7$ and $1c7$.
\end{exam}

We now describe Hilbert space representations for $D_n$. Let $H$ be a
finite-dimensional complex Hilbert space. 
A map $\escript\colon\ascript_n\to H$ satisfying 
$\escript (\cup A_i)=\sum\escript (A_i)$ for any
sequence of mutually disjoint sets $A_i\in\ascript _n$ is a
\textit{vector-valued measure} on $\ascript _n$.  
If $\rmspan\brac{\escript (A)\colon A\in\ascript _n}=H$, then $\escript$
is a \textit{spanning} vector-valued measure. The next result follows
from 
Theorem~2.3 of \cite{gud104} (cf. \cite{djso10}).

\begin{thm}       
\label{thm23}
Let $D_n$ be the $n$-decoherence functional. There exists a spanning vector-valued measure
$\escript\colon\ascript _n\to\complex ^2$ such that $D_n(A,B)=\elbows{\escript (A),\escript (B)}$ for all
$A,B\in\ascript _n$. If $\escript '\colon\ascript _n\to H$ is a spanning
vector-valued measure, then there exists a unitary operator
$U\colon\complex ^2\to H$ such that $U\escript (A)=\escript '(A)$ for
all $A\in\ascript _n$. 
\end{thm}

\begin{cor}       
\label{cor24}
{\rm (a)}\enspace $D_n(\Omega _n,\Omega _n)=1$.
{\rm (b)}\enspace $A\mapsto D_n(A,B)$ is a complex-valued measure for every $B\in\ascript _n$.
{\rm (c)}\enspace If $A_1,\ldots ,A_k$ are sets in $\ascript_n$, 
  then the $k\times k$ matrix $D_n(A_i,A_j)$, $i,j=1,\ldots ,k$ 
  is positive semi-definite (``strong positivity'').
\end{cor}
\begin{proof}
(a)\enspace follows from Theorem~\ref{thm21} and the definition of $D_n$ while 
(b) follows from the definition of $D_n$. To verify 
(c), let $A_1,\ldots ,A_k\in\ascript _n$ and let $a_1,\ldots ,a_k\in\complex$. 
Then by Theorem~\ref{thm23} we have that
\begin{align*}
  \sum _{i,j=1}^kD_n(A_i,A_j)a_i\overline{a_j}
   &=\sum _{i,j=1}^k\elbows{\escript (A_i),\escript (A_j)}a_i\overline{a_j}\\
   &=\elbows{\sum _{i=1}^ka_i\escript (A_i),\sum _{j=1}^ka_j\escript (A_j)}\ge 0\qedhere
\end{align*}
\end{proof}

It follows from Corollary~\ref{cor24}(c) that $\mu _n(A)=D_n(A,A)\ge 0$ for all $A\in\ascript_n$,
$\mu _n(\Omega _n)=1$ and by inspection $\mu _n(\omega )=1/2^n$ for all
$\omega\in\Omega _n$. 
The next result is proved in \cite{gud101, sor94, sor071}.

\begin{thm}       
\label{thm25}
The $n$-truncated $q$-measure $\mu _n$ satisfies the following conditions.
{\rm (a)}\enspace (\textit{grade}-2 \textit{additivity}) For mutually disjoint $A,B,C\in\ascript _n$ we have\newline
$\mu _n(A\cup B\cup C)=\mu _n(A\cup B)+\mu _n(A\cup C)+\mu _n(B\cup C)-\mu _n(A)-\mu _n(B)-\mu _n(C)$
{\rm (b)}\enspace (regularity) If $\mu _n(A)=0$, then $\mu _n(A\cup
B)=\mu _n(B)$ whenever $A\cap B=\emptyset$. If $A\cap B=\emptyset$ and
$\mu _n(A\cup B)=0$, then $\mu _n(A)=\mu _n(B)$. 
\end{thm}

It follows from Theorem~\ref{thm25}(a) and induction that for $3\le m\le n$ we have for
$\brac{i_1,\ldots ,i_m}\subseteq\Omega _n$
\begin{equation}         
\label{eq24}
  \mu _n\paren{\brac{i_1,\ldots ,i_m}} = \sum _{j<k=1}^m\mu _n\paren{\brac{i_j,i_k}}-(m-2)\sum _{j=1}^m\mu _n(i_j)
\end{equation}
We have seen that
\begin{equation*}
\mu _n\paren{\brac{i,j}}=\tfrac{1}{2^{n-1}}+2D_{ij}^n
\end{equation*}
Applying Theorem~\ref{thm21} we conclude that
\begin{equation}         
\label{eq25}
\mu _n\paren{\brac{i,j}}
  =\begin{cases}1/2^{n-1}&\text{if $inj$}\\1/2^{n-2}&\text{if $icj$}\\0&\text{if $idj$}\end{cases}
\end{equation}
An event $A\in\ascript _n$ is \textit{precluded} if $\mu _n(A)=0$. 
The coevent (or anhomomorphic logic) interpretation of the path-integral 
\cite{gt0906, gud102, gud103, sor071, sor072, capetown}
confers a special importance on the precluded events.
We shall show that precluded
events are
relatively
rare. As an illustration, consider Examples~1 and
2. Besides the empty set $\emptyset$ there are no precluded events in
$\ascript _1$ and the only precluded event in 
$\ascript _2$ is $\brac{0,2}$.

If $m\le n$ then we can consider $\Omega _m$ as a subset of $\Omega_n$
by padding on the right with zeros; and this in turn would let us
consider subsets of $\Omega_m$ as subsets of $\Omega_n$.
If this is done,
it follows from \eqref{eq24} and \eqref{eq25} that for $A\in\ascript_m$ we have that
\begin{equation}         
\label{eq26}
  \mu _m(A) = 2^{n-m}\mu _n(A)
\end{equation}
Thus, if $A$ is precluded in $\ascript _m$ then $A$ is precluded in
$\ascript _n$ for all $n\ge m$.
However, this embedding of $\ascript_m$ into $\ascript_n$ is not unique,
nor is it the most natural way to proceed when the elements of
$\ascript_m$ and $\ascript_n$ are thought of as events.
Rather one would regard $\Omega_m$ as a {\it quotient} of $\Omega_n$,
identifying an event in $\Omega_m$ with its lift to $\Omega_n$.  This is
the the point of view adopted implicitly in the following section.


\begin{lem}       
\label{lem26}
If $A\in\ascript _n$ has odd cardinality, then $A$ is not precluded.
\end{lem}
\begin{proof}
Suppose $A=\brac{i_1,\ldots ,i_m}\in\ascript _n$ where $m$ is odd. If $\mu _n(A)=0$ then applying \eqref{eq24} gives
\begin{equation}         
\label{eq27}
\sum _{j<k=1}^m\mu _n\paren{\brac{i_j,i_k}}=\frac{(m-2)m}{2^n}
\end{equation}
where we are assuming that $m\ge 3$ because singleton sets are not precluded. Notice that $(m-2)m$ is odd. However, by \eqref{eq25} the left side of \eqref{eq27} has the form $r/2^n$ where $r$ is even. This is a contradiction. Hence, $\mu _n(A)\ne 0$ so $A$ is not precluded.
\end{proof}

\begin{exam}{5}                    
For $n=3$ we have $\Omega _3=\brac{0,1,\ldots ,7}$. 
Since $\ab{\ascript_3}=2^8$ is large, 
it is impractical to find
by hand
$\mu _3 (A)$ for all $A\in\ascript _3$ so we shall just compute some of
them. Of course, $\mu _3(\emptyset )=0$ and 
$\mu (i)=1/8$, $i=0,1,\ldots ,7$. By \eqref{eq25} we have that
\begin{equation*} 
  \mu _3\paren{\brac{i,j}}
  =\begin{cases}1/4&\text{if $inj$}\\1/2&\text{if $icj$}\\0&\text{if $idj$}\end{cases}
\end{equation*}
Of the 28 doubleton sets the only precluded ones are: 
$\brac{0,2}$, $\brac{0,4}$, $\brac{0,6}$, $\brac{1,5}$, $\brac{3,5}$, $\brac{5,7}$. 
By Lemma~\ref{lem26} there are no precluded tripleton sets  $A=\brac{i,j,k}$. By \eqref{eq24} we have
\begin{equation*}
\mu _3(A)=\mu _3\paren{\brac{i,j}}+\mu _3\paren{\brac{i,k}}+\mu _3\paren{\brac{j,k}}-\tfrac{3}{8}
\end{equation*}
The possibilities are:
$inj, ink, jdk,\mu _3(A)=1/8$;
$inj, ink,jck,\mu _3(A)=5/8$; 
$icj, ick, jck,\mu _3(A)=9/8$ 
and the other possibilities coincide with one of these by symmetry. For
a set of cardinality~4, $A=\brac{i,j,k,l}$ and by \eqref{eq24} we have 
\begin{align*}
  \mu _3(A)
  &=\mu _3\paren{\brac{i,j}}+\mu _3\paren{\brac{i,k}}+\mu _3\paren{\brac{i,l}}+\mu _3\paren{\brac{j,k}}\\
  &\quad +\mu _3\paren{\brac{j,l}}+\mu _3\paren{\brac{k,l}}-1
\end{align*}
The possibilities are:\newline
$inj, ink, inl, jdk, jdl, kcl,\mu _3(A)=1/4$; $inj, ink, inl, jck, jcl,lck,\mu _3(A)\!=\!3/4$;
$inj, ink, idl, jdk, jnl, knl,\mu _3(A)=0$; $inj, ink, idl, jck, jnl, knl,\mu _3(A)=1/2$;
$inj, ink, icl, jck, jnl, knl,\mu _3(A)=1$; $idj, idk, idl, jck, jcl, kcl,\mu (A)=1/2$,
$idj, ick, icl, jdk, jdl, kcl,\mu (A)=1/2$.
The other possibilities coincide with one of these by symmetry. 
Of the 70 sets of cardinality~4 the only precluded ones are:
$\brac{0,2,1,5}$, $\brac{0,2,3,5}$, $\brac{0,2,5,7}$, $\brac{0,4,1,5}$, $\brac{0,4,3,5}$, $\brac{0,4,5,7}$,
$\brac{0,6,1,5}$, $\brac{0,6,3,5}$, $\brac{0,6,5,7}$.
There are no precluded events of cardinality $>4$ in $\Omega_3$.
\end{exam}

\begin{exam}{6}                    
We compute some $q$-measures of events in $\Omega _4=\brac{0,1,\ldots ,15}$. Some precluded doubleton sets are
$\brac{0,2}$, $\brac{0,4}$, $\brac{2,10}$, $\brac{4,10}$. Moreover,
\begin{equation*}
\mu _4\paren{\brac{0,10}}=\mu _4\paren{\brac{2,4}}=1/4
\end{equation*}
It follows from \eqref{eq24} that $\brac{0,2,4,10}$ is precluded.
(In view of Theorem \ref{thm25} (b), this also follows from the fact that 
$\brac{0,2,4,10}$ is the disjoint union of the two precluded sets,
$\brac{0,2}$ and $\brac{4,10}$.)
\end{exam}

\section{Cylinder Sets} 
For $\omega\in\Omega _n$ we identify the pair $(\omega ,0)$ with the
string $\omega 0\in\Omega _{n+1}$ obtained by adjoining $0$ to the right
of the string $\omega$. Similarly we identify $(\omega ,1)$ with 
$\omega 1\in\Omega _{n+1}$. For example, $(011,0)=0110$ and $(011,1)=0111$. 
We can also identify
$\omega\times\brac{0,1}$ with the set $\brac{(\omega ,0),(\omega ,1)}\in\ascript _{n+1}$. In a similar way, for
$A\in\ascript _n$ we define $A\times\brac{0,1}\in\ascript _{n+1}$, by
\begin{equation*}
   A\times\brac{0,1}=\cup\brac{\omega\times\brac{0,1}\colon\omega\in A}
\end{equation*}

\begin{lem}       
\label{lem31}
If $A\in\ascript _n$, then $\mu _{n+1}\paren{A\times\brac{0,1}}=\mu _n(A)$.
\end{lem}
\begin{proof}
For $\omega\in\Omega _n$, let $a(\omega )=i^{c_n(\omega )}$. By \eqref{eq22} we have
\begin{equation*}
D_n(\omega ,\omega ')=\tfrac{1}{2^n}\,a(\omega )\overline{a(\omega ')}\delta _{\alpha _n\alpha '_n}
\end{equation*}
Hence,
\begin{align*}
\mu _{n+1}&\paren{A\times\brac{0,1}}\\
  &=D_{n+1}\paren{A\times\brac{0,1},A\times\brac{0,1}}\\
  &=\sum\brac{D_{n+1}(\omega ,\omega ')\colon\omega ,\omega '\in A\times\brac{0,1}}\\
  &=\sum \brac{D_{n+1}(\omega 0,\omega '0)\colon\omega ,\omega '\in A}\\
    &\quad +\sum\brac{D_{n+1}(\omega 1,\omega '1)\colon\omega ,\omega '\in A}\\
    &=\frac{1}{2^{n+1}}\,\Biggl[
    \sum\brac{a(\omega )\overline{a(\omega ')}\colon\omega ,\omega '\in A,
     \alpha _n,\alpha '_n=0\text{ or }\alpha _n,\alpha '_n=1}\Biggr.\\
    &\quad +i\sum\brac{a(\omega )\overline{a(\omega ')}\colon\omega ,\omega '\in A,\alpha _n=1,\alpha '_n=0}\\
    &\quad -i\sum\brac{a(\omega )\overline{a(\omega ')}\colon\omega ,\omega '\in A,\alpha _n=0,\alpha '_n=1}\\
    &\quad +\sum\brac{a(\omega )\overline{a(\omega ')}\colon\omega ,\omega '\in A,\alpha _n,\alpha '_n=0
      \text{ or }\alpha _n,\alpha '_n=1}\\
    &\quad -i\sum\brac{a(\omega )\overline{a(\omega ')}\colon\omega ,\omega '\in A,\alpha _n=1,\alpha '_n=0}\\
    &\quad +i\sum\Biggl.
    \brac{a(\omega )\overline{a(\omega ')}\colon\omega ,\omega '\in A,\alpha _n=0,\alpha '_n=1}
    \Biggr]\\
    &=\frac{1}{2^n}\sum\brac{a(\omega )\overline{a(\omega ')}\colon\omega ,\omega '\in A,
    \alpha _n,\alpha '_n=0\text{ or }\alpha _n,\alpha '_n=1}\\
    &=\sum\brac{D_n(\omega ,\omega ')\colon\omega ,\omega '\in A}=\mu _n(A)\qedhere
\end{align*}
\end{proof}

We use the notation $A^n=A\times A\cdots\times A$ ($n$ factors).

\begin{cor}       
\label{cor32}
If $A\in\ascript _n$, then $\mu _{n+m}\paren{A\times\brac{0,1}^m}=\mu _n(A)$.
\end{cor}

\begin{cor}       
\label{cor33}
If $A\in\ascript _n$ is precluded, then $A\times\brac{0,1}^m$ is also precluded.
\end{cor}

\REMARK
Lemma \ref{lem31} and its corollaries are valid for any finite unitary
system in the sense of \cite{djsu10}.  Indeed Corollary \ref{cor33} can
be seen as a special case of a much stronger assertion that holds for
such systems:
{\it If $A$ is precluded and $B$ is any subsequent event then the event} 
$C=(A$ and $B)$ 
{\it is also precluded}.
Here, the condition that $B$ be {\it subsequent to} $A$, 
means more precisely the following.  
By definition any event $A$ is a set of histories or paths, and if these
paths are singled out by a condition that concerns their behavior only
for times $t<t_0$, we will say that $A$ is ``earlier than'' $t_0$.
Defining events later than $t_0$ analogously, we then say that $B$ is
subsequent to $A$ if for some $t_0$, $A$ is earlier than $t_0$ and $B$
is later.  The event $(A$ and $B)$ is of course simply the intersection
$A\cap B$ expressed as a logical conjunction.  We can also write it as
$(A$ and-then $B)$
in order to emphasize that
$B$ is meant to be subsequent to $A$.
The preservation of preclusion by `and-then' can be viewed as
a kind of causality condition.  
When this condition is fulfilled, one can
correlate to any
event $B$ later than $t_0$ and earlier than $t_1$ 
a linear operator 
from the Hilbert space associated with times $t<t_0$ 
to that associated with times $t<t_1$~.
In a situation like that of Cor.~\ref{cor33}, 
the earlier (resp. later) Hilbert space would be 
that associated to $\ascript_n$ (resp. $\ascript_{n+m}$). 

\medskip

\begin{exam}{7}                    
In Example 2 we saw that $\brac{0,2}\in\ascript _2$ is precluded. Applying Corollary~\ref{cor33} shows that
$\brac{0,4,1,5}\in\ascript _3$ is also precluded. Applying Corollary~\ref{cor33} again shows that
$\brac{0,8,2,10,1,9,3,11}\in\ascript _4$ is precluded.
\end{exam}

Using our previously established notation we can write $\Omega =\brac{0}\times\brac{0,1}\times\cdots$ or
$\Omega =\Omega _n\times\brac{0,1}\times\brac{0,1}\times\cdots$. A subset $A\subseteq\Omega$ is a
\textit{cylinder set} if there exists a $B\in\ascript _n$ for some $n\in\positive$ such that
$A=B\times\brac{0,1}\times\brac{0,1}\times\cdots$. Thus, the first $n+1$
bits for strings in $A$ are restricted and the further bits are not. For
$\omega\in\Omega _n$, we call 
$\rmcyl (\omega )=\brac{\omega}\times\brac{0,1}\times\brac{0,1}\times\cdots$
an \textit{elementary cylinder set}.\footnote
{In references \cite{gfrjr02, gfrjr03, sor111} the term cylinder-set is reserved for what are
 here called {\it elementary} cylinder sets.}
If $\omega =\alpha_0\alpha _1\cdots\alpha _m\in\Omega _m$ and
$\omega '=\alpha _0\alpha _1\cdots\alpha _m\alpha _{m+1}\cdots\alpha _n\in\Omega _n$, $m\le n$ we say that
$\omega'$ is an \textit{extension} of $\omega$. We have that $\rmcyl
(\omega ')\subseteq\rmcyl (\omega )$ if and only if $\omega '$ is an
extension of $\omega$ and $\rmcyl (\omega ')\cap\rmcyl (\omega)=\emptyset$ 
if and only if neither 
$\omega$ or $\omega '$ is an extension of the other. 
(Thus any two {elementary cylinder sets} are either disjoint or nested.)
Moreover, any cylinder set is a finite disjoint union of elementary cylinder sets.

We denote the collection of all cylinder sets by $\cscript (\Omega )=\cscript$. If $A\in\cscript$ then its complement $A'$ is clearly in $\cscript$. Similarly, $\cscript$ is closed under finite unions and finite intersections so $\cscript$ is an algebra of subsets of $\Omega$. Of course, there are subsets of $\Omega$ that are not in $\cscript$; for example,
$\brac{\omega}\notin\cscript$ for $\omega\in\Omega$. For $A\in\cscript$ if
$A=B\times\brac{0,1}\times\brac{0,1}\times\cdots$ with $B\in\ascript _n$ we define $\mu (A)=\mu _n(B)$. To show that $\mu\colon\cscript\to\real ^+$ is well-defined, suppose $A=B_1\times\brac{0,1}\times\brac{0,1}\times\cdots$ with $B_1\in\ascript _m$. If $m=n$, then $B=B_1$ and we're finished. Otherwise, we can assume that $m<n$. It follows that $B=B_1\times\brac{0,1}^{n-m}$. Hence, $\mu _n(B)=\mu _m(B_1)$ by Corollary~\ref{cor32} so $\mu$ is well-defined. It is clear that $\mu\colon\cscript\to\real ^+$ satisfies Conditions (a) and ( b) of Theorem~\ref{thm25} so we can consider $\mu$ as a $q$-measure on $\cscript$.

As before, we say that $A\in\cscript$ is \textit{precluded} if $\mu (A)=0$. We also say that $B\in\cscript$ is
\textit{stymied} if $B\subseteq A$ for some precluded $A\in\cscript$. Of course a precluded set is stymied but there are many stymied sets that are not precluded. For instance, by Example~2, $\rmcyl (000)$ and $\rmcyl (101)$ are not precluded but are stymied. It is clear that $\mu (\Omega )=1$ and $\Omega$ is not stymied. Surprisingly, it is shown in \cite{djsu10} that $\Omega$ is the only set in $\cscript$ that is not stymied.

Let $A_1\supseteq A_2\supseteq\cdots$ be a decreasing sequence in $\cscript$ with $A=\cap A_i\in\cscript$. (In general, $A$ need not be in $\cscript$.) We shall show in the proof of Theorem~4.1 that $\Omega$ is compact in the product topology and that every element of $\cscript$ is compact. Letting
$B_i=A_i\smallsetminus A$, since $B_i\in\cscript$ we conclude that $B_i$ is compact, $i=1,2,\ldots$, and that $B_1\supseteq B_2\supseteq \cdots$ with $\cap B_i=\emptyset$. It follows that $B_m=\emptyset$ for some
$m\in\positive$. Hence, $A_m=A$ so $A_i=A$ for $i\ge m$. We conclude that
\begin{equation}         
\label{eq31}
\lim\mu (A_i)=\mu\paren{\cap A_i}
\end{equation}
Now let $A_1\subseteq A_2\subseteq\cdots$ be an increasing sequence in $\cscript$ with $\cup A_i\in\cscript$. By taking complements of our previous work we have
\begin{equation}         
\label{eq32}
\lim _{i\to\infty} \mu (A_i)=\mu\paren{\cup A_i}
\end{equation}
Since $\mu$ satisfies \eqref{eq31} and \eqref{eq32} we say that $\mu$ is \textit{continuous} on $\cscript$.

Let $\ascript$ be the $\sigma$-algebra generated by $\cscript$. If $\mu$
were a finitely additive probability measure on $\cscript$ satisfying
\eqref{eq31} or \eqref{eq32} , then by the Kolmogorov extension theorem,
$\mu$ 
would have
a unique extension to a (countably additive) probability
measure on $\ascript$. The next example shows that this extension
theorem does not hold for $q$-measures; that is, $\mu$ does not have an
extension to a continuous $q$-measure on 
$\ascript$.

\begin{exam}{8}                    
Let
\begin{equation*} 
B_1=\brac{0010,0100,0110}=\brac{2,4,6}\in\ascript _3
\end{equation*}
As in Example~5 $\mu _3(B_1)=9/8$. Letting $B_2=\brac{010,100,110}$ we have that
\begin{equation*} 
B_1\times B_2=\brac{0010,0100,0110}\times\brac{010,100,110}\in\ascript _6
\end{equation*}
A simple calculation shows that $\mu _6 (B_1\times B_2)=(9/8)^2$ and continuing,
$\mu _9 (B_1\times B_2\times B_2)=(9/8)^3$. Defining $A_i\in\cscript$ by
$A_1=B_1\times\brac{0,1}\times\brac{0,1}\times\cdots$,
$A_2=B_1\times B_2\times\brac{0,1}\times\brac{0,1}\times\cdots$,
$A_3=B_1\times B_2\times B_2\times\brac{0,1}\times\brac{0,1}\times\cdots$ we have that
$A_1\supseteq A_2\supseteq\cdots$. However, $\mu (A_i)=(9/8)^i$ so $\lim _{i\to\infty}\mu (A_i)=\infty$. Hence, if
$\mu$ had an extension to $\ascript$, then \eqref{eq31} would fail.
\end{exam}

Another way to show that $\mu$ does not extend to $\ascript$ is given in \cite{djsu10}. We define the
\textit{total variation} $\ab{\mu}$ of $\mu$ by 
\begin{equation*} 
\ab{\mu}(A)=\sqbrac{\sup _{\pi (A)}\sum _i\mu (A_i)^{1/2}}^2
\end{equation*}
for all $A\in\cscript$ where the supremum is over all finite partitions $\pi (A)=\brac{A_1,\ldots ,A_n}$ of $A$ with
$A_i\in\cscript$. We say that $\mu$ is of \textit{bounded variation} if $\ab{\mu (A)}<\infty$ for all $A\in\cscript$. It is shown in \cite{djsu10} that if $\mu$ has an extension to a continuous $q$-measure on $\ascript$, then $\mu$ must be of bounded variation. It is proved in \cite{djsu10} that for any finite unitary system, $\mu$ is not of bounded variation. Although this proof is difficult for our particular case it is simple.

\begin{exam}{9}                    
We show that $\mu$ is not of bounded variation. For $0,1,\ldots ,2^n-1\in\Omega _n$ we have the partition of
$\Omega$
\begin{align*}
\Omega&=\bigcup _{i=1}^{2^n-1}\rmcyl (i)\\
\intertext{and}
\sum _{i=0}^{2^n-1}\mu\sqbrac{\rmcyl (i)}^{1/2}&=\sqrt{2^n\,}
\end{align*}
Hence, 
$\ab{\mu}(\Omega )\ge 2^n$ 
for all $n\in\positive$ so
$\ab{\mu}(\Omega )=\infty$. A similar argument shows that
$\ab{\mu}(A)=\infty$ for all $A\in\cscript$, $A\ne\emptyset$. 
\end{exam}

Although we cannot extend $\mu$ to a continuous $q$-measure on
$\ascript$, perhaps we can extend $\mu$ to physically interesting sets
in $\ascript\smallsetminus\cscript$. 
We now discuss a possible way to accomplish this.\footnote
{Some of the ideas expressed here in embryo are developed further in \cite{sor111}.}
For $\omega =\alpha _0\alpha _1\cdots\in\Omega$ and
$A\subseteq\Omega$ we write $\omega (n)A$ if there is an 
$\omega '=\beta _0\beta _1\cdots\in A$ such that $\beta _i=\alpha _i$,
$i=0,1,\ldots ,n$. 
We then define 
\begin{equation*} 
     A^{(n)} = \brac{\omega\in\Omega\colon\omega (n)A}
\end{equation*}
Notice that 
$A^{(n)}\in\cscript$, 
$A^{(0)}\supseteq A^{(1)}\supseteq A^{(2)}\supseteq\cdots$, and
$A\subseteq\cap A^{(n)}$. 
We think of $A^{(n)}$ as 
a particular sort of
time-$n$ cylindrical approximation to $A$. 
We say that
$A\subseteq\Omega$ is  
a \textit{lower set} 
if $A=\cap A^{(n)}$;
and we call $A$ \textit{beneficial} if
$\lim\mu\paren{A^{(n)}}$ exists and is finite. 
%
%
We denote the collection of lower 
sets by $\lscript$, the collection of
beneficial sets by $\bscript$ and write 
$\bscript_{\lscript}=\bscript\cap\lscript$. 
If $A\in\bscript$, 
we define $\muhat (A)=\lim\mu\paren{A^{(n)}}$.

The next section considers algebraic structures but for now we mention that Example~9 to follow shows that
$\lscript$ is not closed under ${}'$ so is not an algebra. Since $\ascript$ is closed under countable intersections,
$\lscript\subseteq\ascript$. If $A\in\cscript$, then $A=A^{(n)}=A^{(n+1)}=\cdots$ for some $n\in\positive$. Hence,
$A=\cap A^{(n)}$ and $\mu (A)=\lim\mu\paren{A^{(n)}}=\muhat (A)$. Thus,
$\cscript\subseteq\bscript _{\lscript}$ and the definition of $\muhat$
on $\bscript _{\lscript}$ reduces to the usual definition of $\mu$ on
$\cscript$. The following result shows that $\brac{\omega}\in\bscript
_{\lscript}$ and $\muhat\paren{\brac{\omega}}=0$ for all 
$\omega\in\Omega$. 
We conclude that $\bscript _{\lscript}$ properly contains $\cscript$.

\begin{lem}       
\label{lem34}
If $A\subseteq\Omega$ with $\ab{A}<\infty$, then $A\in\bscript _{\lscript}$ and $\muhat (A)=0$.
\end{lem}
\begin{proof}
Suppose that $\omega =\alpha _0\alpha _1\cdots\notin A$. Then there exists an $n\in\positive$ such that
$\alpha _0\alpha _1\cdots\alpha _n$ is different from the first $n$ bits of all $\omega '\in A$. But then
$\omega\notin A^{(n)}$ so $\omega\notin\cap A^{(n)}$. Hence, $A=\cap A^{(n)}$. If $\ab{A}=m$, then
$A^{(n)}=B_n\times\brac{0,1}\times\brac{0,1}\times\cdots$, $B_n\in\ascript _n$ with $\ab{B_n}\le m$,
$n=0,1,2,\ldots\,$. Hence
\begin{equation*} 
\mu (A^{(n)})=\mu _n(B_n)=D_n(B_n,B_n)=\sum\brac{D_n(\omega ,\omega ')\colon\omega ,\omega '\in B_n}
  \le\tfrac{m^2}{2^n}
\end{equation*}
Hence, $\lim\mu\paren{A^{(n)}}=0$. We conclude that $A\in\bscript _{\lscript}$ and $\muhat (A)=0$.
\end{proof}

\begin{exam}{10}                    
Let $A\subseteq\Omega$ with $\ab{A}<\infty$, $A\ne\emptyset$. We then have that ${A'}^{(n)}=\Omega$,
$n=0,1,2,\ldots\,$. Hence, $A'\ne\cap{A'}^{(n)}=\Omega$ so $A'\notin\lscript$. This shows that $\lscript$ is not an algebra. This also shows that $\bscript _{\lscript}$ is not an algebra.
\end{exam}

\begin{exam}{11}                    
Define the set
\begin{equation*}
A=\brac{\omega\in\Omega\colon\omega\text{ has at most one }1}
\end{equation*}
We have that
\begin{equation*}
  A^{(n)}=\brac{00\cdots 0,010\cdots 0,0010\cdots 0,\cdots ,00\cdots 01}
\end{equation*}
It is clear that $A=\cap A^{(n)}$. Also, we have
\begin{equation*}
\mu\paren{A^{(n)}}=\tfrac{1}{2^n}\sqbrac{n+1-2(n+1)+2\binom{n-1}{2}}=\tfrac{1}{2^n}(n^2-4n+5)
\end{equation*}
Hence, $\lim\mu\paren{A^{(n)}}=0$ so $A\in\bscript _{\lscript}$ and $\muhat (A)=0$.
\end{exam}

\section{Quadratic Algebras} 
This section discusses algebraic structures for the collections
$\lscript$, $\bscript$ and $\bscript _{\lscript}$. A collection $Q$ of
subsets of a set $S$ is a \textit{quadratic algebra} if $\emptyset, S\in
Q$ and if $A,B,C\in Q$ are mutually disjoint and $A\cup B,A\cup C,B\cup
C\in Q$, then $A\cup B\cup C\in Q$. If $Q$ is a quadratic algebra,  
a $q$-\textit{measure} on $Q$ is a map $\nu\colon Q\to\real ^+$ such that if $A,B,C\in Q$ are mutually disjoint and
$A\cup B,A\cup C,B\cup C\in Q$, then
\begin{equation*}
\nu(A\cup B\cup C)=\nu (A\cup B)+\nu (A\cup C)+\nu (B\cup C)-\nu (A)-\nu (B)-\nu (C)
\end{equation*}

\begin{exam}{12}                    
Let $S=\brac{d_1,d_2,d_3,u_1,u_2,u_3,s_1,s_2,s_3}$ and define $Q\subseteq 2^S$ by $\emptyset,S\in Q$ and
$A\in Q$ if and only if each of the three types of elements have different cardinalities in $A$, $A\ne\emptyset ,S$. For instance,
\begin{equation*}
\brac{u_1,d_1,d_2},\brac{u_1,d_1,d_2,s_1,s_2,s_3}\in Q
\end{equation*}
and these are the only kinds of sets in $Q$ besides $\emptyset ,S$. Although $Q$ is closed under complementation, it is not closed under disjoint unions so $Q$ is not an algebra. For instance $\brac{u_1,d_1,d_2},\brac{u_2,s_1,s_2}\in Q$ but
\begin{equation*}
\brac{u_1,u_2,d_1,d_2,s_1,s_2}\notin Q
\end{equation*}
To show that $Q$ is a quadratic algebra, suppose $A,B,C\in Q$ are mutually disjoint and
$A\cup B,A\cup C,B\cup C\in Q$. If one or more of $A,B,C$ are empty then clearly,
$A\cup B\cup C\in Q$ so suppose $A,B,C\ne\emptyset$. Since $\ab{A}=\ab{B}=\ab{C}=3$, we have
$A\cup B\cup C=S\in Q$. An example of a $q$-measure on $Q$ is $\nu (\emptyset )=0$, $\nu (S)=1$, $\nu (A)=1/6$ if $\ab{A}=3$ and $\nu (A)=1/2$ if $\ab{A}=6$. If $A,B,C$ are mutually disjoint nonempty sets in $Q$, then
$\nu (A\cup B\cup C)=\nu (\Omega )=1$ and
\begin{equation*}
\nu(A\cup B)+\nu (A\cup C)+\nu (B\cup C)-\nu (A)-\nu (B)-\nu (C)=\tfrac{3}{2}-\tfrac{1}{2}=1
\end{equation*}
Notice that $\nu$ is not additive because
\begin{equation*}
\nu\paren{\brac{u_1,d_1,d_2}}+\nu\paren{\brac{u_2,u_3,s_1}}=\tfrac{1}{3}\ne\tfrac{1}{2}
  =\nu\paren{\brac{u_1,u_2,u_3,d_1,d_2,s_1}}
\end{equation*}
\end{exam}

\begin{exam}{13}                    
Let $S=\brac{x_1,\ldots ,x_n,y_1,\ldots ,y_m}$ where $n$ is odd. Let
\begin{equation*}
Q=\brac{A\subseteq S\colon\ab{\brac{x_i\colon x_i\in A}}=0\text{ or odd}}
\end{equation*}
Notice that $Q$ is not closed under complementation, unions (even disjoint unions) or intersections. To show that
$Q$ is a quadratic algebra, suppose $A,B,C\in Q$, are mutually disjoint and $A\cup B,A\cup C,B\cup C\in Q$. Since $A\cup C,B\cup C\in Q$ at most one of $A,B,C$ has an odd number of $x_i$s and the other contain no $x_i$s. Hence, $A\cup B\cup C\in Q$. An example of a nonadditive $q$-measure on $Q$ is
$\nu (A)=\ab{A}^2$. In fact, $\nu (A)=\ab{A}^2$ is a $q$-measure on any finite quadratic algebra.
\end{exam}

\begin{thm}       
\label{thm41}
$\lscript$ and $\bscript _{\lscript}$ are quadratic algebras and $\muhat$ is a $q$-measure on $\bscript _{\lscript}$ that extends $\mu$ on $\cscript$.
\end{thm}
\begin{proof}
Placing the discrete topology on $\brac{0,1}$, since $\brac{0,1}$ is compact, by Tychonov's theorem
$\Omega =\brac{0,1}\times\brac{0,1}\times\cdots$ is compact in the product topology. Any cylinder set is closed
(and open) and hence is compact. Let $A,B\in\lscript$ with $A\cap B=\emptyset$. Since
$A=\cap A^{(n)}$, $B=\cap B^{(n)}$ we have that
\begin{equation*}
\cap\paren{A^{(n)}\cap B^{(n)}}=\paren{\cap A^{(n)}}\cap\paren{\cap B^{(n)}}=A\cap B=\emptyset
\end{equation*}
Since $A^{(n)}\cap B^{(n)}$ is a decreasing sequence of compact sets
with empty intersection, 
there exists an
$n\in\positive$ such that $A^{(m)}\cap B^{(m)}=\emptyset$ for $m\ge n$. 
Now let $A,B,C\in\lscript$ be mutually disjoint with $A\cup B, A\cup C,B\cup C\in\lscript$. 
By our previous work there exists an
$n\in\positive$ such that 
$A^{(m)},B^{(m)},C^{(m)}$ are mutually disjoint for $m\ge n$. 
By the distributive law we have
\begin{align*}
  A\cup B\cup C&=\paren{\cap A^{(m)}}\cup\paren{\cap B^{(m)}}\cup\paren{\cap C^{(m)}}\\
  &=\cap\paren{A^{(m)}\cup B^{(m)}\cup C^{(m)}}=\cap\sqbrac{(A\cup B\cup C)^{(m)}}
\end{align*}
Hence, $\lscript$ is a quadratic algebra. If $A,B,C,A\cup B,A\cup C,B\cup C\in\bscript _{\lscript}$ with $A,B,C$ disjoint, since $A^{(m)},B^{(m)},C^{(m)}$ are eventually disjoint we conclude that
\begin{align*}
\lim\mu&\sqbrac{(A\cup B\cup C)^{(m)}}=\lim\mu\sqbrac{A^{(m)}\cup B^{(m)}\cup C^{(m)}}\\
  &=\lim\mu\paren{A^{(m)}\cup B^{(m)}}+\lim\mu\paren{A^{(m)}\cup C^{(m)}}+\lim\paren{B^{(m)}\cup C^{(m)}}\\
  &\quad -\lim\mu (A^{(m)})-\lim\mu (B^{(m)})-\lim\mu (C^{(m)})\\
  &=\muhat (A\cup B)+\muhat (A\cup C)+\muhat (B\cup C)-\muhat (A)-\muhat (B)-\muhat (C)
\end{align*}
Hence, $A\cup B\cup C\in\bscript _{\lscript}$ so $\bscript _{\lscript}$ is a quadratic algebra. Also,
\begin{equation*}
\muhat (A\cup B\cup C)=\muhat (A\cup B)+\muhat (A\cup C)+\muhat (B\cup C)-\muhat (A)-\muhat (B)-\muhat (C)
\end{equation*}
so $\muhat$ is a $q$-measure on $\bscript _{\lscript}$ that extends $\mu$ on $\cscript$.
\end{proof}

We say that $A\subseteq\Omega$ is an \textit{upper set} 
if $A=\cup A^{\prime (n)\prime}$ 
and denote the collection of upper sets by $\uscript$. 
Since $A^{'(n)}$ is a decreasing sequence of cylinder sets, 
we conclude if $A\in\uscript$ 
then $A$ is the union of an increasing sequence of cylinder sets $A^{'(n)'}$. 
For example, if $\ab{A}<\infty$ we have shown that $A\in\lscript$ so that $A=\cap A^{(n)}$. 
Hence, 
\begin{equation*}
  A'=\cup A^{(n)\prime}=\cup (A')^{\prime (n)\prime}
\end{equation*}
It follows that $A'\in\uscript$ so $\uscript$ properly contains $\cscript$. 
Moreover, $A'\notin\lscript$ so $\uscript\notin\lscript$.

\begin{lem}       
\label{lem42}
Suppose $B\subseteq\Omega$ and there exists a decreasing sequence $C_i\in\cscript$ and an increasing sequence $D_i\in\cscript$ such that
\begin{equation*}
B=\cap C_i=\cup D_i
\end{equation*}
Then $B\in\cscript$.
\end{lem}
\begin{proof}
We have that $D_i\subseteq B\subseteq C_i$, $C_i\smallsetminus D_i\in\cscript$ and
\begin{equation*}
\cap (C_i\smallsetminus D_i)=\cap (C_i\cap D'_i)=\paren{\cap C_i}\cap\paren{\cap D'_i}=B\cap B'=\emptyset
\end{equation*}
Since $C_i\smallsetminus D_i$ is compact in the product topology, there exists a $j\in\positive$ such that
$C_j\smallsetminus D_j=\emptyset$. Therefore, $D_j=C_j$. Since $D_j\subseteq B\subseteq C_j$, we have
$B=D_j=C_j$ so that $B\in\cscript$.
\end{proof}

\begin{cor}       
\label{cor43}
{\rm (a)}\enspace $\lscript\cap\uscript =\cscript$.
{\rm (b)}\enspace If $A,A'\in\lscript$, then $A\in\cscript$.
\end{cor}
\begin{proof}
(a)\enspace This follows directly from Lemma~\ref{lem42}.
(b)\enspace Since $A\in\lscript$, we have that $A=\cap A^{(n)}$ where $A^{(n)}\in\cscript$ is decreasing. Since
$A'\in\lscript$ we have that $A'=\cap A^{\prime (n)}$. Hence, $A=\cup A^{\prime (n)\prime}$ where
$A^{\prime (n)\prime}\in\cscript$ is increasing. By Lemma~\ref{lem42}, $A\in\cscript$.
\end{proof}

\begin{thm}       
\label{thm44}
{\rm (a)}\enspace If $A\in\lscript$, then $A'\in\uscript$.
{\rm (b)}\enspace $\uscript$ is a quadratic algebra.
\end{thm}
\begin{proof}
(a)\enspace If $A\in\lscript$, then $A=\cap A^{(n)}$ so that $A'=\cup (A')^{\prime (n)\prime}$. Hence,
$A'\in\uscript$.
(b)\enspace Clearly, $\emptyset,\Omega\in\uscript$. Suppose $A,B,C\in\uscript$ are mutually disjoint. Since
\begin{equation*}
(A\cup B\cup C)^{\prime (n)}=(A'\cap B'\cap C')^{(n)}\subseteq (A')^{(n)}\cap (B')^{(n)}\cap (C')^{(n)}
\end{equation*}
we have that
\begin{equation*}
A^{\prime (n)\prime}\cup B^{\prime (n)'}\cup C^{\prime (n)\prime}\subseteq (A\cup B\cup C)^{\prime (n)\prime}
\end{equation*}
Hence,
\begin{align*}
A\cup B\cup C&=\paren{\cup A^{\prime (n)\prime}}\cup\paren{\cup B^{\prime (n)\prime}}
  \cup\paren{\cup C^{\prime (n)\prime}}\\
  &=\cup\paren{A^{\prime (n)\prime}\cup B^{\prime (n)\prime}\cup C^{\prime (n)\prime}}
  \subseteq\cup (A\cup B\cup C)^{\prime (n)\prime}
\end{align*}
But $(A\cup B\cup C)^{\prime (n)\prime}\subseteq A\cup B\cup C$ so that
\begin{equation*}
A\cup B\cup C=\cup (A\cup B\cup C)^{\prime (n)\prime}
\end{equation*}
Therefore, $A\cup B\cup C\in\uscript$ so $\uscript$ is a $q$-algebra.
\end{proof}

Letting
\begin{equation*}
\bscript _{\uscript}=\brac{A\in\uscript\colon\lim\mu _n(A^{\prime (n)\prime})\text{ exists}}
\end{equation*}
we see that $\bscript _{\uscript}$ is the ``upper'' counterpart of $\bscript _{\lscript}$. As before, if
$A\in\bscript _{\uscript}$ we define $\muhat (A)=\lim\mu _n(A^{\prime (n)\prime})$.
Unfortunately, we have not been able to show that $\bscript _{\uscript}$ is a quadratic algebra. However, we shall show that $\brac{\gamma}'\in\bscript _{\uscript}$ for $\gamma\in\Omega$. We first need the following lemma.

\begin{lem}       
\label{lem45}
For $n\in\positive$, $j=0,1,\ldots ,2^n-1$, the function $c_n(j)$ satisfies
\begin{equation*}
c_{n+1}(2^{n+1}-1-j)=c_n(j)+1
\end{equation*}
\end{lem}
\begin{proof}
Let $j\in\Omega _n=\brac{0,1,\ldots ,2^n-1}$ and for $a\in\brac{0,1}$, let $a'=a+1\pmod{2}$. If $j$ has binary representation $j=a_0a_1\cdots a_n$, $a_0=0$, $a_k\in\brac{0,1}$, $k=1,\ldots ,n$, since
\begin{equation*}
a_0a_1\cdots a_n+a'_0a'_1\cdots a'_n=2^{n+1}-1
\end{equation*}
we have that
\begin{equation*}
(2^{n+1}-1)-j=0a'_0a'_1\cdots a'_n\in\Omega _{n+1}
\end{equation*}
Suppose that $c_n(j)=k$ so $a_0a_1\cdots a_n$ has $k$ position switches. These position switches are in
one-to-one correspondence with the position switches in $a'_0a'_1\cdots a'_n$. Since $a'_0=1$,
$0a'_0a'_1\cdots a'_n$ has one more position switch so
\begin{equation*}
c_{n+1}(2^{n+1}-j-1)=k+1\qedhere
\end{equation*}
\end{proof}

\begin{exam}{14}                    
Since $c_1=(0,1)$, it follows immediately from Lemma~\ref{lem45} that $c_2=(0,1,2,1)$, $c_3=(0,1,2,1,2,3,2,1)$ and 
\begin{equation*}
c_4=(0,1,2,1,2,3,2,1,2,3,4,3,2,3,2,1)
\end{equation*}
\end{exam}

To show that $\brac{\gamma}'\in\uscript$, for simplicity let $\gamma =000\cdots$ and let $B=\brac{\gamma}'$.

\begin{thm}       
\label{thm46}
The set $B\in\bscript _{\uscript}$ and $\muhat (B)=1$.
\end{thm}
\begin{proof}
For ease of notation, let $B_n=B^{\prime (n)\prime}=\brac{\gamma}^{(n)\prime}$. Then
$B_1=\brac{00}'\times\brac{0,1}\times\cdots$, $B_2=\brac{000}\times\brac{0,1}\times\cdots$, $\cdots$ and
$B=\cup B_n\in\uscript$. We must show that $\lim\mu _n(B_n)=1$. From the definition of $B_n$ we have that
\begin{align}         
\label{eq41}
\mu _n(B_n)&=\sum\brac{D_{\gamma ,\gamma '}^n\colon\gamma ,\gamma '\ne 0}
  =1-\sum\brac{D_{\gamma ,\gamma '}^n\colon\gamma\text{ or }\gamma '=0}\notag\\
  &=1-\frac{1}{2^n}-\frac{1}{2^{n-1}}\,\sum _{j=2}^{2^n-2}\brac{i^{c_n(j)}\colon j\text{ even}}\notag\\
  &=1-\frac{1}{2^n}-\frac{1}{2^{n-1}}\,\sum _{j=1}^{2^{n-1}-1}i^{c_n(2j)}
  =1+\frac{1}{2^n}-\frac{1}{2^{n-1}}\,\sum _{j=0}^{2^{n-1}-1}i^{c_n(2j)}
\end{align}
Let $u_n(j)$ be the number of $j$-values of $c_n$. For example $u_3(0)=1$, $u_3(1)=3$, $u_3(2)=3$, $u_3(3)=1$. It follows from Lemma~\ref{lem45} that

\begin{equation}         
\label{eq42}
u_{n+1}(j)=u_n(j)+u_n(j-1),\quad j=1,2,\ldots ,n+1
\end{equation}
Letting
\begin{equation*}
v_n(j)=\sum\brac{u_n(k)\colon k=j\pmod{4}}
\end{equation*}
for $j=0,1,2,3$ we have that
\begin{equation*}
v_n(j)=u_n(j)+u_n(j+4)+u_n(j+8)+\cdots +u_n\paren{4\bigg\lfloor\frac{n-j}{4}\bigg\rfloor+j}
\end{equation*}
where $\lfloor{x}\rfloor$ is the largest integer less than or equal to $x$. Applying \eqref{eq42} we conclude that $v_n$ satisfies the recurrence relations
\begin{equation}         
\label{eq43}
v_{n+1}(j)=v_n(j)+v_n(j-1)
\end{equation}
for $j=1,2,3,4$ where $v_n(-1)=v_n(3)$. Also, $v_n$ satisfies the initial conditions $v_1(0)=v_1(1)=1$,
$v_1(2)=v_1(3)=0$.

We now prove by mathematical induction on $n$ that
\begin{equation}         
\label{eq44}
v_n(j)=2^{n-2}+2^{\frac{n}{2}-1}\cos (n-2j)\pi /4
\end{equation}
By the initial conditions, \eqref{eq44} holds for $n=1$, $j=0,1,2,3$. Suppose \eqref{eq44} holds for $n$ and $j=0,1,2,3$. We then have by \eqref{eq43} that 
\begin{align*}
v_{n+1}(j)&=v_n(j)+v_n(j-1)\\
  &=2^{n-1}+2^{\frac{n}{2}-1}\sqbrac{\cos (n-2j)\pi /4+\cos\paren{n-2(j-1)}\pi /4}\\
  &=2^{n-1}+2^{\frac{n}{2}-1}\sqbrac{\cos (n-2j)\pi /4-\sin (n-2j)\pi /4}\\
  &=2^{n-1}+2^{\frac{n}{2}-1}2^{1/2}\cos\sqbrac{(n-2j)\pi /4+\pi /4}\\
  &=2^{(n+1)-2}-2^{\frac{n+1}{2}-1}\cos\sqbrac{\paren{(n+1)-2j}\pi /4}\\
\end{align*}
This proves \eqref{eq44} by induction.

Applying \eqref{eq41} we have that
\begin{equation}         
\label{eq45}
\mu _n(B_n)=1+\tfrac{1}{2^n}-\tfrac{1}{2^{n-1}}\sqbrac{v_n(0)-v_n(2)}
\end{equation}
By \eqref{eq44} we have
\begin{align*}
v_n(0)&=2^{n-1}+2^{\frac{n}{2}-1}\cos n\pi /4\\
\intertext{and}
v_n(2)&=2^{n-2}+2^{\frac{n}{2}-1}\cos (n-4)\pi /4=2^{n-2}-2^{\frac{n}{2}-1}\cos n\pi /4
\end{align*}
Hence, \eqref{eq44} becomes
\begin{equation*}
\mu _n(B_n)=1+\frac{1}{2^n}-\frac{2^{n/2}}{2^{n-1}}\,\cos n\pi /4
  =1+\frac{1}{2^n}-\frac{1}{2^{n/2-1}}\,\cos n\pi /4
\end{equation*}
We conclude that $\lim\mu _n(B_n)=1$.
\end{proof}

Notice that $C=\brac{0111\cdots}'$ is the event that the particle ever returns to the site $0$. Analogous to
Theorem~\ref{thm46} we have that $C\in\bscript _{\uscript}$ and $\mu (C)=1$. Thus, $C$ is a physically significant event in $\bscript _{\uscript}\smallsetminus\cscript$ and the $q$-probability of return in unity.

We say that $A,B\subseteq\Omega$ are \textit{strongly disjoint} if there is an $n\in\positive$ such that
$A^{(n)}\cap B^{(n)}=\emptyset$. It is clear that we then have that $A^{(m)}\cap B^{(m)}=\emptyset$ for $m\ge n$. Since $A\subseteq A^{(n)}$, $B\subseteq B^{(n)}$, if $A$ and $B$ are strongly disjoint, then $A\cap B=\emptyset$. However, the converse does not hold.

\begin{exam}{15}                    
Define $A\subseteq\Omega$ by
\begin{equation*}
A=\brac{\omega\in\Omega\colon\omega\text{ has finitely many 1s}}
\end{equation*}
Then $A\cap A'=\emptyset$ but $A^{(n)}=A^{\prime (n)}=\Omega$ for all $n\in\positive$. Hence,
$A^{(n)}\cap A^{\prime (n)}\ne\emptyset$ so $A,A'$ are disjoint but not strongly disjoint.
\end{exam}

A collection of subsets $Q\subseteq 2^\Omega$ is a \textit{weak quadratic algebra} if $\emptyset ,\Omega\in Q$ and if $A,B,C\in Q$ are strongly disjoint and $A\cup B,A\cup C,B\cup C\in Q$, then $A\cup B\cup C\in Q$. If $Q$ is a weak quadratic algebra, a $q$-\textit{measure} on $Q$ is a map $\nu\colon Q\to\real ^+$ such that if $A,B,C\in Q$ are strongly disjoint and $A\cup B,A\cup C,B\cup C\in Q$, then
\begin{equation*}
\nu (A\cup B\cup C)=\nu (A\cup B)+\nu (A\cup C)+\nu (B\cup C)-\nu (A)-\nu (B)-\nu (C)
\end{equation*}
The proof of the next theorem is similar to that of Theorem~\ref{thm41}.

\begin{thm}       
\label{thm47}
$\bscript$ and $\bscript _{\uscript}$ are weak quadratic algebras and $\muhat$ is a $q$-measure on $\bscript$ that extends $\mu$ to $\bscript$.
\end{thm}

Let $A$ be the set in Example~15. We have that $A,A'\in\bscript$ and $\muhat (A)=\muhat (A')=1$. Since
$A,A'\notin\lscript$, we see that $\bscript$ properly contains $\bscript _{\lscript}$. We also have that the set $\bscript$ of Theorem~\ref{thm46} is in $\bscript$ but not in $\lscript$.

\section{``Expectations''}  
This section explores the mathematical analogy between functions on
$\Omega$ and random variables in classical probability theory, using a
notion of ``expectation'' introduced in \cite{gud09, gud104}.  
When applied to the characteristic function $\chi_A$ of an event
$A\in\ascript$, this
expectation reproduces the quantum measure $\mu(A)$ of $A$, which
classically would be the probability that the event $A$ occurs.  One
knows that such an interpretation is not viable quantum mechanically when
interference is present; and one must seek elsewhere for the physical
meaning of $\mu$ \cite{capetown}.
We hope that the formal relationships we expose here can be helpful in this
quest, or in making further contact with the more traditional quantum
formalism. 

The following paragraphs consider expectations 
in terms
of a $q$-integral \cite{gud09, gud104}. 
For a 
positive
random variable
$f\colon\Omega _n\to\real ^+$ we define 
\begin{align}         
\label{eq51}
\int fd\mu _n&=\sum _{i,j=0}^{2^n-1}\min\sqbrac{f(\omega _i),f(\omega _j)}D_n(\omega _i,\omega _j)\notag\\
  &=\sum _{i,j=0}^{2^n-1}\min\sqbrac{f(\omega _i),f(\omega _j)}D_{ij}^n
\end{align}
An arbitrary random variable $f\colon\Omega _n\to\real$ has a unique representation $f=f^+-f^-$ where
$f^+,f^-\ge 0$ and $f^+f^-=0$ and we define
\begin{equation*}
\int fd\mu _n=\int f^+d\mu _n-\int f^-d\mu _n
\end{equation*}
This $q$-integral has the following properties. If $f\ge 0$, then $\int fd\mu _n\ge 0$,
$\int\alpha fd\mu _n=\alpha\int fd\mu _n$ for all $\alpha\in\real$, $\int\chi _Ad\mu _n=\mu _n(A)$ for all
$A\in\ascript _n$ where $\chi _A$ is the characteristic function of $A$. However, in general
\begin{equation*}
\int (f+g)d\mu _n\ne\int fd\mu _n+\int gd\mu _n
\end{equation*}

\begin{thm}       
\label{thm51}
If $a_1,\ldots ,a_n\in\real ^+$, then the matrix $M_{ij}=\sqbrac{\min (a_i,a_j)}$ is positive semi-definite.
\end{thm}
\begin{proof}
We can assume without loss of generality that $a_1\le a_2\le\cdots\le a_n$. We then write
\begin{equation}         
\label{eq52}
M=\left[\begin{matrix}\noalign{\smallskip}a_1&a_1&a_1&\cdots&a_1\\\noalign{\smallskip}
  a_1&a_2&a_2&\cdots&a_2\\\noalign{\smallskip}a_1&a_2&a_3&\cdots&a_3\\\noalign{\smallskip}
  \vdots&&&&\\\noalign{\smallskip}a_1&a_2&a_3&\cdots&a_n\\\noalign{\smallskip}
  \end{matrix}\right]
\end{equation}
Subtracting the first column from the other columns gives the determinant
\begin{equation}         
\label{eq53}
\ab{M}=\left[\begin{matrix}\noalign{\smallskip}a_1&0&0&\cdots&0\\\noalign{\smallskip}
  a_1&a_2-a_1&a_2-a_1&\cdots&a_2-a_1\\\noalign{\smallskip}
  a_1&a_2-a_1&a_3-a_1&\cdots&a_3-a_1\\\noalign{\smallskip}
  \vdots&&&&\\\noalign{\smallskip}a_1&a_2-a_1&a_3-a_1&\cdots&a_n-a_1\\\noalign{\smallskip}
  \end{matrix}\right]
\end{equation}
We now prove by induction on $n$ that
\begin{equation}         
\label{eq54}
\ab{M}=a_1(a_2-a_1)(a_3-a_2)\cdots (a_n-a_{n-1})
\end{equation}
For $n=1$ we have $M=\sqbrac{a_1}$ and $\ab{M}=a_1$ and for $n=2$ we have
\begin{equation*}
M=\left[\begin{matrix}\noalign{\smallskip}a_1&a_1\\\noalign{\smallskip}
  a_1&a_2\\\noalign{\smallskip}\end{matrix}\right]
\end{equation*}
and $\ab{M}=a_1a_2-a_1^2=a_1(a_2-a_1)$. Suppose the result \eqref{eq54} holds for $n-1$ and let $M$ have the form \eqref{eq52}. Then $\ab{M}$ has the form \eqref{eq53} so we have
\begin{equation*}
\ab{M}=a_1\left[\begin{matrix}\noalign{\smallskip}a_2-a_1&a_2-a_1&a_2-a_1&\cdots&a_2-a_1\\
  \noalign{\smallskip}a_2-a_1&a_3-a_1&a_3-a_1&\cdots&a_3-a_1\\\noalign{\smallskip}
  a_2-a_1&a_3-a_1&a_4-a_1&\cdots&a_4-a_1\\\noalign{\smallskip}
  \vdots&&&&\\\noalign{\smallskip}a_2-a_1&a_3-a_1&a_4-a_1&\cdots&a_n-a_1\\\noalign{\smallskip}
  \end{matrix}\right]
\end{equation*}
It follows from the induction hypothesis that
\begin{align*}
\ab{M}&=a_1(a_2-a_1)\sqbrac{(a_3-a_1)-(a_2-a_1)}\sqbrac{(a_4-a_1)-(a_3-a_1)}\\
  &\qquad\cdots\sqbrac{(a_n-a_1)-(a_{n-1}-a_1)}\\
  &=a_1(a_2-a_1)(a_3-a_2)\cdots (a_n-a_{n-1})
\end{align*}
This completes the induction proof. Since $a_1\le a_2\le\cdots\le a_n$, we conclude that $\ab{M}\ge 0$. Since all the principal submatrices of $M$ have the form \eqref{eq52}, they also have nonnegative determinants. Hence, $M$ is positive semi-definite.
\end{proof}

If $f\colon\Omega _n\to\real ^+$, define the $2^n\times 2^n$ matrix $\fhat$ given by
\begin{equation*}
\fhat _{ij}=\min\sqbrac{f(i),f(j)}
\end{equation*}
It follows from Theorem~\ref{thm51} that $\fhat$ is positive semi-definite. For $f\colon\Omega _n\to\real$ we can write $f=f^+-f^-$ in the canonical way where $f^+,f^-\ge 0$. Define the self-adjoint matrix $\fhat$ by
\begin{equation*}
\fhat _{ij}=f_{ij}^{+\wedge}-f_{ij}^{-\wedge}
\end{equation*}
Applying \eqref{eq51} we have that
\begin{equation}         
\label{eq55}
\int fd\mu _n=\sum _{i,j=0}^{2^n-1}\fhat _{ij}D_{ij}^n
\end{equation}
We 
might
think of $\fhat$ as the 
``observable''
representing the 
``random variable''
$f$. The next result shows that
$\int fd\mu _n$ is 
then
given by the usual quantum formula for the expectation of the observable $\fhat$ in the 
``state''
$D^n$.

\begin{thm}       
\label{thm52}
For $f\colon\Omega _n\to\real$ we have that $\int fd\mu _n=\rmtr (\fhat D^n)$.
\end{thm}
\begin{proof}
Applying \eqref{eq55}, since $\fhat$ is 
symmetric
we have
\begin{equation*}
\int fd\mu _n=\sum _{i,j}\fhat _{ij}D_{ij}^n=\sum _{i,j}\fhat _{ji}D_{ij}^n
  =\sum _j(\fhat D^n)_{jj}=\rmtr (\fhat D^n)\qedhere
\end{equation*}
\end{proof}

\begin{cor}       
\label{cor53}
For any $A\in\ascript _n$ we have that $\mu _n(A)=\rmtr (\chihat _AD^n)$.
\end{cor}
\begin{proof}
It follows from Theorem~\ref{thm52} that
\begin{equation*}
\mu _n(A)=\int\chi _Ad\mu _n=\rmtr (\chihat _AD^n)\qedhere
\end{equation*}
\end{proof}

It is also interesting to note that $\chihat _A=\ket{\chi _A}\bra{\chi _A}$ so we can write
$\mu _n(A)=\rmtr\paren{\ket{\chi _A}\bra{\chi _A}D^n}= \bra{\chi_A} D^n \ket{\chi_A}$.  
More generally, we have
\begin{equation}         
\label{eq56}
  D_n(A,B) = \rmtr\paren{\ket{\chi _A}\bra{\chi _B}D^n} = \bra{\chi_B} D^n \ket{\chi_A}
\end{equation}
and $(A,B)\mapsto\ket{\chi _A}\bra{\chi _B}$ is a positive semi-definite 
operator-bimeasure. 
That is, 
it is an
operator-valued measure in each variable and $A_1,\ldots ,A_m\subseteq\Omega _n$,
$c_1,\ldots ,c_m\subseteq \complex$ 
imply that
\begin{equation*}
\sum _{i,j}c_1\overline{c_j}\ket{\chi _{A_i}}\bra{\chi _{A_j}}\ge 0
\end{equation*}
Although $(f+g)^\wedge\ne\fhat +\ghat$ in general, the proof of the following lemma is straightforward.

\begin{lem}       
\label{lem54}
If $f,g,h\colon\Omega _n\to\real$ have disjoint support, then
\begin{equation*}
(f+g+h)^\wedge =(f+g)^\wedge +(f+h)^\wedge +(g+h)^\wedge -\fhat -\ghat -\hhat
\end{equation*}
\end{lem}

Applying Lemma~\ref{lem54} and Theorem~\ref{thm52} gives the following result.

\begin{cor}       
\label{cor55}
If $f,g,h\colon\Omega _n\to\real$ have disjoint support, then
\begin{align*}
\int (f+g+h)d\mu _n&=\int (f+g)d\mu _n+\int (f+h)d\mu _n+\int (g+h)d\mu _n\\
  &\qquad -\int fd\mu _n-\int gd\mu _n-\int hd\mu _n
\end{align*}
\end{cor}

The next theorem can be used to simplify computations.

\begin{thm}       
\label{thm56}The eigenvalues of $D^n$ are $1/2$ with multiplicity 2 and 0 with multiplicity $2^n-2$. The unit eigenvectors corresponding to $1/2$ are $\psi _0^n,\psi _1^n$ where
\begin{equation*}
\psi _0^n=\frac{1}{2^{(n-1)/2}}\left[\begin{matrix}\noalign{\smallskip}i^{c_n(0)}\\\noalign{\smallskip}
  0\\\noalign{\smallskip}i^{c_n(2)}\\\noalign{\smallskip}0\\\noalign{\smallskip}\vdots\\\noalign{\smallskip}
  i^{c_n(2^n-2)}\\\noalign{\smallskip}0\\\noalign{\smallskip}\end{matrix}\right]\,,\qquad
  \psi _1^n=\frac{1}{2^{(n-1)/2}}\left[\begin{matrix}\noalign{\smallskip}0\\\noalign{\smallskip}
  i^{c_n(1)}\\\noalign{\smallskip}0\\\noalign{\smallskip}i^{c_n(3)}\\\noalign{\smallskip}0\\\noalign{\smallskip}
  \vdots\\\noalign{\smallskip}0\\\noalign{\smallskip}i^{c_n(2^n-1)}\\\noalign{\smallskip}\end{matrix}\right]
\end{equation*}
\end{thm}
\begin{proof}
Applying \eqref{eq23} we have for $j$ odd that
\begin{equation*}
D^n\psi _0^n(j)=0=\tfrac{1}{2}\,\psi _0^n(j)
\end{equation*}
and for $j$ even that
\begin{align*}
D^n\psi _0^n(j)&=\frac{1}{2^{(3n-1)/2}}\,\sum\brac{i^{\sqbrac{c_n(j)-c_n(k)}}i^{c_n(k)}\colon k\text{ even}}\\
  &=\frac{1}{2^{(3n-1)/2}}\,i^{c_n(j)}2^{n-1}=\frac{1}{2}\,\frac{i^{c_n(j)}}{2^{(n-1)/2}}=\frac{1}{2}\,\psi _0^n(j)
\end{align*}
Hence, $D^n\psi _0=\tfrac{1}{2}\psi _0$ and a similar argument shows that $D^n\psi _1=\tfrac{1}{2}\psi _1$. Thus, $1/2$ is an eigenvalue with unit eigenvectors $\psi _0^n,\psi _1^n$. Now the $k$th column of $D_n$ for $k>0$ and $k$ even is the vector
\begin{align*}
\left[\tfrac{1}{2^n}\,\right.&\left.i^{-c_n(k)}i^{c_n(j)}p_{jk},j=0,1,\ldots ,2^n-1\right]\\
&=\frac{i^{-c_n(k)}}{2^{(n+1)/2}}\,\sqbrac{\frac{1}{2^{(n-1)/2}}i^{c_n(j)}p_{jk},j=1,2,\ldots ,2^n-1}\\
&=\frac{i^{-c_n(k)}}{2^{(n+1)/2}}\,\psi _0^n
\end{align*}
Thus, the $k$th column of $D^n$ for $k>0$ and even is a multiple of $\psi _0^n$ and similarly, the $k$th column of $D^n$ for $k>1$ and odd is a multiple of $\psi _1^n$. Hence, the range of $D^n$ is generated by $\psi _0^n$ and
$\psi _1^n$. 
Thus, $\rmnull (D^n)=\rmspan\brac{\psi _0^n,\psi _1^n}^\perp$ 
so $0$ is an eigenvalue of $D^n$ with multiplicity $2^n-2$.
\end{proof}

It follows from Theorem~\ref{thm56} that
\begin{equation}         
\label{eq57}
  D^n=\tfrac{1}{2}\,\ket{\psi _0^n}\bra{\psi _0^n}+\tfrac{1}{2}\ket{\psi _1^n}\bra{\psi _1^n}
\end{equation}
Applying \eqref{eq56} and \eqref{eq57} gives
\begin{align*}
  D_n(A,B)&=\tfrac{1}{2}\,\elbows{\chi _A,\psi _0^n}\elbows{\psi _0^n,\chi _B}
  +\tfrac{1}{2}\,\elbows{\chi _A,\psi _1^n}\elbows{\psi _1^n, \chi _B}\\
\intertext{and}
\mu _n(A)&=\tfrac{1}{2}\,\ab{\elbows{\chi _A,\psi _0^n}}^2+\tfrac{1}{2}\,\ab{\elbows{\chi _A,\psi _1^n}}^2
\end{align*}
Also, if $f\colon\Omega _n\to\real$, then by Theorem~\ref{thm52} we have
\begin{equation}         
\label{eq58}
\int fd\mu _n=\rmtr (\fhat D_n)=\tfrac{1}{2}\,\elbows{\fhat\psi _0^n,\psi _0^n}
  +\tfrac{1}{2}\elbows{\fhat\psi _1^n,\psi _1^n}
\end{equation}
We close by computing some expectations. Let $f_n\colon\Omega _n\to\real ^+$ be the random variable given by
\begin{equation*}
f_n(\omega _i)=\text{ number of }1\text{s in }\omega _i
\end{equation*}
The proof of the next result is
similar to that of Lemma~\ref{lem45}. 

\begin{lem}       
\label{lem57}
For $n\in\positive$, $j=0,1,\ldots ,2^n-1$, the function $f_n(j)$ satisfies
\begin{equation*}
f_{n+1}(j+2^n)=f_n(j)+1
\end{equation*}
\end{lem}

\begin{exam}{16}                    
Since $f_1=(0,1)$, it follows from Lemma~\ref{lem57} that $f_2=(0,1,1,2)$, $f_3=(0,1,1,2,1,2,2,3)$ and
\begin{equation*}
f_4=(0,1,1,2,1,2,2,3,1,2,2,3,2,3,3,4)
\end{equation*}
\end{exam}

\begin{exam}{17}                    
Applying \eqref{eq58} we have that
\begin{align*}
\int f_1d\mu _1&=1/2,\quad\int f_2d\mu _2=3/2,\quad\int f_3d\mu _3=2\\
\intertext{and}
\int c_1d\mu _1&=1/2,\quad\int c_2d\mu _2=3/2,\quad\int c_3d\mu _3=3
\end{align*}
Unfortunately, it appears to be difficult to find general formulas for $\int f_nd\mu _n$ and $\int c_nd\mu _n$.
\end{exam}


\begin{thebibliography}{99}
\bibitem{gfrjr02}
Graham Brightwell, {H. Fay Dowker}, {Raquel S. Garc{\'\i}a}, {Joe Henson} and {Rafael D.~Sorkin},
``General Covariance and the `Problem of Time' in a Discrete Cosmology,'' in K.G.~Bowden, Ed.,     
 {\it Correlations}, Proceedings of the ANPA 23 conference, held August 16-21, 2001,
 Cambridge, England (Alternative Natural Philosophy Association, London, 2002), pp 1-17,
{gr-qc/0202097}
{http://www.perimeterinstitute.ca/personal/rsorkin/some.papers/} 
\bibitem{gfrjr03}
Graham Brightwell, Fay Dowker, Raquel S.~Garc{\'\i}a, Joe Henson and Rafael D.~Sorkin,
``Observables in Causal Set Cosmology,''
\journaldata {Phys. Rev.~D} {67} {084031} {2003},
{gr-qc/0210061}
{http://www.perimeterinstitute.ca/personal/rsorkin/some.papers/} 
\bibitem{gt0906}Yousef Ghazi-Tabatabai, 
Quantum measure: A new interpretation, arXiv: quant-ph (0906:0294).
\bibitem{djso10}Fay Dowker, Steven Johnston and Rafael D. Sorkin, 
Hilbert spaces from path integrals, arXiv: quant-ph (1002:0589), 2010.
\bibitem{djsu10}Fay Dowker, Steven Johnston, Sumati Surya,
On extending the quantum measure, arXiv: quant-ph (1002:2725), 2010.
\bibitem{gud09}S.~Gudder, Quantum measure and integration theory, \textit{J. Math. Phys.} \textbf{50},
123509 (2009).
\bibitem{gud101}S.~Gudder, Quantum measure theory, \textit{Math. Slovaca} \textbf{60}, 681--700 (2010).
\bibitem{gud102}S.~Gudder, An anhomomorphic logic for quantum mechanics, \textit{J. Phys. A}
\textbf{43}, 095302 (2010).
\bibitem{gud103}S.~Gudder, Quantum reality filters, \textit{J. Phys. A} \textbf{43}, 48530 (2010).
\bibitem{gud104}S.~Gudder, Hilbert space representations of decoherence functionals and quantum measures,
arXiv: quant-ph (1011.1694) 2010.
\bibitem{mocs05}Xavier Martin, Denjoe O'Connor and Rafael D.~Sorkin,
The Random Walk in Generalized Quantum Theory,
\textit{Physic Rev D} \textbf{71}, 024029 (2005).
\bibitem{sor94}Rafael D.~Sorkin,
Quantum mechanics as quantum measure theory, \textit{Mod. Phys. Letts.~A} \textbf{9} (1994), 3119--3127.
\bibitem{sor071}Rafael D. Sorkin,
Quantum dynamics without the wave function, \textit{J.~Phys.~A} \textbf{40} (2007), 3207-3231.
\bibitem{sor072}Rafael D. Sorkin,
An exercise in ``anhomomorphic logic'', \textit{J.~Phys.: Conference Series} (JPCS) 67,012018 (2007).
\bibitem{sor111}Rafael~D.~Sorkin,
``Toward a `fundamental theorem of quantal measure theory'\.'' (to appear)
{http://arxiv.org/abs/1104.0997}, 
{http://www.perimeterinstitute.ca/personal/rsorkin/some.papers
/141.fthqmt.pdf} 
\bibitem{capetown}Rafael~D.~Sorkin, 
``Logic is to the quantum as geometry is to gravity,'' in G.F.R. Ellis, J. Murugan and A. Weltman (eds),
 {\it Foundations of Space and Time} (Cambridge University Press), (to appear)
 {arXiv:1004.1226 [quant-ph]},
 {http://www.perimeterinstitute.ca/personal/rsorkin/some.papers/} 


\end{thebibliography}
\end{document}